\documentclass{llncs}

\usepackage[utf8x]{inputenc}
\usepackage{amsfonts} 
\usepackage{amssymb} 
\usepackage{fontenc} 
\usepackage{graphicx} 
\usepackage{bbm}
\usepackage{mathtools}
\usepackage{tikz}
\usetikzlibrary{matrix}
\PassOptionsToPackage{hyphens}{url}\usepackage{hyperref}

\title{Adapting Real Quantifier Elimination Methods for Conflict Set
  Computation} \author{Maximilian Jaroschek\inst{1} \and Pablo Federico
  Dobal\inst{1,2,3} \and Pascal Fontaine\inst{3}\thanks{This work has been
    supported by the ANR/DFG project STU 483/2-1 SMArT, project ANR-13-IS02-0001
    of the Agence Nationale de la Recherche, by the European Union Seventh
    Framework Programme under grant agreement no. 295261 (MEALS), by the
    R\'egion Lorraine, and by the STIC AmSud MISMT} } \institute{ Max Planck
  Institute for Informatics, Saarbr\"ucken, Germany \and Universität des
  Saarlandes, Saarbr\"ucken, Germany \and INRIA, Universit\'e de Lorraine \&
  LORIA, Nancy, France }

\let\set\mathbbm
\begin{document}

\maketitle

\begin{abstract}
  The satisfiability problem in real closed fields is decidable.  In the context
  of satisfiability modulo theories, the problem restricted to conjunctive sets
  of literals, that is, sets of polynomial constraints, is of particular
  importance.  One of the central problems is the computation of good
  explanations of the unsatisfiability of such sets, i.e.\ obtaining a small
  subset of the input constraints whose conjunction is already unsatisfiable.
  We adapt two commonly used real quantifier elimination methods, cylindrical
  algebraic decomposition and virtual substitution, to provide such conflict
  sets and demonstrate the performance of our method in practice.

\end{abstract}

\section{Introduction}

Among the reasons for the current success of Satisfiability Modulo Theory (SMT,
we refer to~\cite{Barrett14} for more information) solvers is the ability to
handle large formulas in an expressive language.  Since arithmetic is pervasive
in applications of SMT, this language should include some kind of arithmetic
theory.  Linear arithmetic (on reals and integers) was one of the first theories
considered for SMT~\cite{Nelson3}, and integrated in practice into SMT
solvers~\cite{Barrett5,Detlefs2}.  Non-linear arithmetic is also mentioned in
the fundamental combination of theories paper~\cite{Nelson3}.  Although many
applications do require non-linear arithmetic reasoning --- our motivating
application was the verification of a clock synchronization
algorithm~\cite{Barsotti2} --- it is considered in practice only since quite
recently (e.g.~\cite{Jovanovic3}), and few solvers integrate non-linear
arithmetic reasoning capabilities.  Up to now, no technique is accepted as the
right way to integrate non-linear reasoning capabilities into SMT
solvers.

The theory of real closed fields (reals with order, addition, and
multiplication) has however been extensively studied in the area of symbolic
computation, and mature tools exist to handle sets of constraints in this
language, e.g.~\cite{dolzmann2,brown2}.  The results presented
here aim at adapting those tools so that they can be integrated into an SMT
framework.  Indeed, whereas developing dedicated techniques for non-linear
arithmetic within SMT is crucial, a lesson from linear arithmetic is that mature
(external) tools should also be adapted for cooperation with SMT solvers.  For
instance, a reasonably efficient linear programming tool suitably incorporated
into the SMT solver CVC4 provided an impressive improvement of efficiency
compared to the dedicated SMT techniques alone~\cite{King1}.

To integrate a theory reasoner in an SMT framework, some features are valuable
(see Section~1.4.1 in~\cite{Barrett14}).  Since we envision fast and incomplete
techniques tightly integrated within SMT, backed up by a complete and robust but
also heavy engine, it is not of foremost importance for this engine to be
incremental and backtrackable: it will only be called as a last resort on a full
assignment when the heuristic solver failed to show unsatisfiability.  However,
a critical feature is that the complete engine provides models, both for
feedback to the user but also for model-based combination with other
theories~\cite{Moura5,Caminha4}.  Adapting established real closed field
decision procedures to produce models has been the subject of a previous
work~\cite{Kosta1}.  The other critical feature is to be able, from an
unsatisfiable set of constraints, to extract a small conflict set.  Without this
ability, the cooperation of the SMT solver and the engine would most probably
fail because the SMT solver would enumerate an exponential number of slightly
different assignments, successively submitted to the engine.  The engine would
reject them one by one, but they would essentially be unsatisfiable for the same
reason.  With small conflict set production, all these assignments are
blocked by the strong conflict clause added within the SMT solver in just one
call to the external engine.

We here focus on the computation of small conflict sets from unsatisfiable sets
of non-linear constraints.  Two commonly used real quantifier elimination
methods, namely cylindrical algebraic decomposition and virtual substitution,
are considered.  They basically share a feature that provides the key to
efficiently compute conflict sets: a finite set of test points is generated in
the process.  These test points falsify some of the input constraints.  If the
tentative conflict set contains enough constraints so that at least one of them
is false for each test point, it is indeed a conflict set.

Section~\ref{sec:qe} briefly describes the two decision procedures for sets of
polynomial constraints on the reals, Section~\ref{sec:ic} presents the small
conflict set extraction method, and experimental results are discussed in
Section~\ref{sec:exp}.

\section{Real Quantifier Elimination}
\label{sec:qe}

Given a quantified formula $\phi$, quantifier elimination is the
process of finding an equivalent, quantifier-free formula
$\phi'$. Whether or not quantifier elimination is possible in theory
and practice in general depends on the considered formal system and
the underlying theory.

For first-order logic formulas over the reals it is well known that quantifier
elimination is possible.  This was first proven by Tarski in 1951~\cite{tarski},
but the first successful algorithmic approach to the problem was developed by
Collins in 1974~\cite{collins}.  To formally define the problem, consider a
quantifier-free first-order formula $\varphi(x_1,\dots,x_n,u_1,\dots,u_m)$ over
the reals in the variables $x_1,\dots,x_n$, $u_1,\dots,u_m$. Given the formula
\begin{equation*}
\phi\equiv Q_1 x_1,\dots Q_n x_n:\varphi(x_1,\dots,x_n,u_1,\dots,u_m),
\end{equation*}
with $Q_i\in\{\forall,\exists\}$ for $1\leq i\leq n$, the quantifier-elimination
problem consists in finding a quantifier-free first-order formula
$\phi'(u_1,\dots,u_m)$ such that $\phi'$ is logically equivalent to $\phi$.  It
was proven independently by Weispfenning~\cite{weispfenning2} and Davenport and
Heintz~\cite{davenport} that solving the quantifier elimination problem over
real closed fields can require double exponential space.

Subsequently we describe two widely used real quantifier elimination
methods. Both approaches are based on the same general idea which we
discuss first before going into details about the specifics for
each method. Our goal is to give a comprehensible and intuitive
introduction to these procedures and not to describe them in thorough
technical detail.  References to more in depth treatments of the
subjects are given for the interested reader.

While these methods work in a general context, our focus lies on input
formulas found in the SMT setting with only existential quantifiers
and no free variables:
\begin{equation}
\label{eq:form2}
\phi\equiv \exists x_1,\dots \exists x_n:\varphi(x_1,\dots,x_n),
\end{equation}
It is clear that then either
$\texttt{true}$ or $\texttt{false}$ is a quantifier-free equivalent of
$\phi$. Over the reals, quantifier-free formulas are Boolean
combinations of polynomial expressions of the form
$p(x_1,\dots,x_n)\bowtie 0$ where $p$ is a polynomial in
$\set R[x_1,\dots,x_n]$ and $\bowtie$ is a relation symbol in
$\{<,\leq,=,\neq,>,\geq\}$. Given a point $(a_1,\dots,a_n)\in\set
R^n$, we can see if $\varphi$ holds for this point by substituting
$a_i$ for $x_i$ for all $1\leq i\leq n$. If we were able to perform
the substitution for all points in $\set R^n$ in finite time, we could
easily see if $\phi$ holds or not.

The approach of the two quantifier elimination methods
\textit{cylindrical algebraic decomposition} (CAD) and \textit{virtual
  substitution} (VS) is to reduce the set of infinitely many points in
$\set R^n$ to a finite set of test points, i.e. to find a finite
subset~$T$ of $\set R^n$ such that $\phi$ holds over $\set R^n$ if and only
if it holds over $T$.

\subsection{Cylindrical Algebraic Decomposition}

Cylindrical algebraic decomposition~\cite{collins} is the most
widely used real quantifier elimination method to date.  It is based on a simple
observation: given a finite, non-empty set~$P$ of polynomials in $n$ variables,
one can define an equivalence relation on $\set R^n$ that decomposes the space
into finitely many connected cells such that all the given polynomials are sign
invariant in each cell.

\begin{definition}
  Let $P$ be a non-empty set of polynomials in $\set R[x_1,\dots,x_n]$. For 
  $a,b\in\set R^n$ we say that $a$ is equivalent to $b$ if there exists a path
  $\gamma:[0,1]\rightarrow \set R^n$ from $a$ to $b$ such that for all
  $s,t\in[0,1]$ and all $p\in P$ we have that
  \[\operatorname{sgn}(p(\gamma(s))) =  \operatorname{sgn}(p(\gamma(t))).\]
  The term \emph{cell} refers to the preimage of an equivalence class
  under the canonical homomorphism which maps a point to its
  equivalence class. We call the set of all cells an (algebraic)
  decomposition of $\set R^n$.
\end{definition}

\begin{example}
\label{ex:cad:1}
  To illustrate the basic idea, we consider the bivariate case, and the
  following set of polynomials.
  \[P=\{\underbrace{x^2+y^2-1}_{p_1},\underbrace{x^2-y+1/2}_{p_2}\}\]
  The first polynomial defines three connected, sign invariant cells in
  $\set R^2$ given by
  \[\{(a,b)\in\set R^2\mid p_1<0\},\{(a,b)\in\set R^2\mid
  p_1=0\},\{(a,b)\in\set R^2\mid p_1>0\},\] and similarly, $p_2$ also
  decomposes $\set R^2$ into three cells when not taking $p_1$ into
  account. The combination of the cells induced by $p_1$ and the cells
  induced by $p_2$ gives rise to a new decomposition where the
  original cells either persist, collapse into common cells or form
  new cells via intersection. The decomposition of $\set R^2$ induced
  by $P$ consists of 5 different cells in total, as illustrated
  in Figure~\ref{fig:cad}.
\end{example}

\begin{figure}
\begin{center}
\begin{tikzpicture}[scale=1.3,line cap=round,line join=round,x=1.0cm,y=1.0cm]
\draw[->,color=black] (-2.92,0) -- (3.09,0);
\foreach \x in {-2.5,-2,-1.5,-1,-0.5,0.5,1,1.5,2,2.5,3}
\draw[shift={(\x,0)},color=black] (0pt,2pt) -- (0pt,-2pt);
\draw[->,color=black] (0,-1.49) -- (0,3.22);
\foreach \y in {-1,-0.5,0.5,1,1.5,2,2.5,3}
\draw[shift={(0,\y)},color=black] (2pt,0pt) -- (-2pt,0pt);
\clip(-2.92,-1.49) rectangle (3.09,3.22);
\draw [line width=1pt] (0,0) circle (1cm);
\clip(-2.41,-1.76) rectangle (2.33,2.79);
\draw [line width=1pt] (0,0) circle (1cm);
\draw [samples=100,xshift=0cm,yshift=0.5cm,line width=1pt] plot[domain=0:3]
(\x,\x^2);
\draw [samples=100,xshift=0cm,yshift=0.5cm,line width=1pt] plot[domain=-3:0] (\x,-\x^2);
\draw (-1.88,1.26) node[anchor=north west] {1};
\draw (-0.4,0.3) node[anchor=north west] {2};
\draw (-0.19,0.94) node[anchor=north west] {3};
\draw (0.34,2.46) node[anchor=north west] {4};
\draw (1.37,1.69) node[anchor=north west] {5};
\draw [dotted, line width=1.4pt] (1.11,1.73)-- (1.41,1.58);
\end{tikzpicture}
\end{center}
\caption{\label{fig:cad}The sign invariant cells of Example~\ref{ex:cad:1}. Note
  that cell no. 5 is given by the union of the varieties of $p_1$ and $p_2$.}
\end{figure}
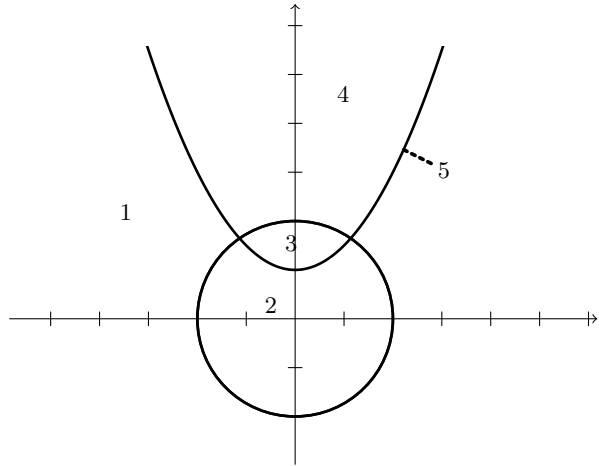

To study a quantified formula $\phi$, we want to  collect in a set
$P$ all the polynomial expressions in $\phi$ and then compute a sample point for
each cell in the decomposition induced by $P$.  While it seems easy to identify
the different sign invariant cells simply by inspection of the plot of the
varieties in Figure~\ref{fig:cad}, it is a non-trivial task for a computer and
for more involved polynomial systems (in more than two variables).

To facilitate the algorithmic identification of different cells, new polynomials
are added to $P$ so that the decomposition becomes cylindrical in the following
sense:

\begin{definition}
  A decomposition of\/ $\set R^n$ is called cylindrical if $n=1$ or if there
  exists a projection $\pi: \set R^n\rightarrow\set R^{n-1}$ that acts on the
  elements of\/ $\set R^n$ by removing one of
  their coordinates such that the following two conditions hold:
  \begin{enumerate}
  \item For two cells $C_1,C_2\subset\set R^n$, either $\pi(C_1)=\pi(C_2)$
    or $\pi(C_1)\cap\pi(C_2)=\emptyset$.
  \item The decomposition of $\set R^{n-1}$ induced by the images
    under $\pi$ of the cells in the decomposition of $\set R^{n}$ is
    cylindrical.
  \end{enumerate}
  We call a set of polynomials $P\subset R[x_1,\dots,x_n]$ cylindrical
  if the decomposition of\/ $\set R^n$ induced by $P$ is cylindrical.
\end{definition}
Again, this can easily be illustrated by an example.

\begin{example} (Example~\ref{ex:cad:1} continued.)
\label{ex:cad:2}
  The decomposition induced by $P$ as in Example~\ref{ex:cad:1} is not
  cylindrical. We can, however, refine it by adding four linear
  polynomials to the set. Let $c = \sqrt{0.5(\sqrt{7}-2)}$ ($c$ is
  such that $p_1(\pm c)=p_2(\pm c)$) and set
  \[P'=P\cup\{x+1,x+c,x-c,x-1\}.\] $P'$ is cylindrical and the
  decomposition is illustrated in Figure~\ref{fig:cad2}. It consists of 47
  different cells.

  Starting from a set of sample points from each cell in the induced
  decomposition of $\set R$ (represented by the dots on the horizontal axis in
  the figure), we can easily find all cells in $\set R^2$ ``above'' a fixed cell
  in $\set R$ by keeping the $x_1$ value fixed and looking for roots of any
  polynomial in $P$ with that $x_1$ value. In the picture, this corresponds to
  moving along the dotted line and looking for sign changes.
  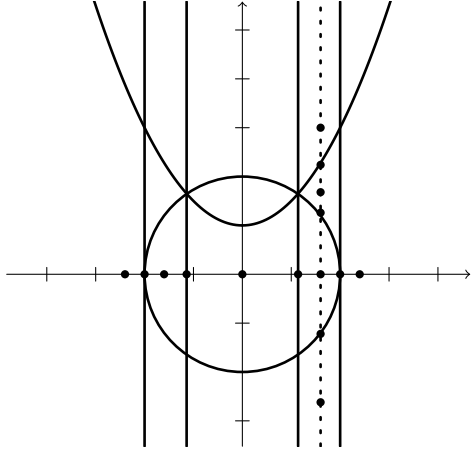
\begin{figure}[ht]
  \begin{center}
\begin{tikzpicture}[scale=1.3,line cap=round,line join=round,x=1.0cm,y=1.0cm]
\draw[->,color=black] (-2.41,0) -- (2.33,0);
\foreach \x in {-2,-1.5,-1,-0.5,0.5,1,1.5,2}
\draw[shift={(\x,0)},color=black] (0pt,2pt) -- (0pt,-2pt);
\draw[->,color=black] (0,-1.76) -- (0,2.79);
\foreach \y in {-1.5,-1,-0.5,0.5,1,1.5,2,2.5}
\draw[shift={(0,\y)},color=black] (2pt,0pt) -- (-2pt,0pt);
\clip(-2.41,-1.76) rectangle (2.33,2.79);
\draw [line width=1pt] (0,0) circle (1cm);
\draw [samples=100,xshift=0cm,yshift=0.5cm,line width=1pt] plot[domain=0:3]
(\x,\x^2);
\draw [samples=100,xshift=0cm,yshift=0.5cm,line width=1pt] plot[domain=-3:0] (\x,-\x^2);
\draw [line width=1pt] (1,-1.76) -- (1,2.79);
\draw [line width=1pt] (-1,-1.76) -- (-1,2.79);
\draw [line width=1pt] (0.57,-1.76) -- (0.57,2.79);
\draw [line width=1pt] (-0.57,-1.76) -- (-0.57,2.79);
\draw [loosely dotted,line width=1pt, ] (0.8,-1.76) -- (0.8,2.79);
\begin{scriptsize}
\fill  (0,0) circle (1.2pt);
\fill  (0.57,0) circle (1.2pt);
\fill  (0.8,0) circle (1.2pt);
\fill (1,0) circle (1.2pt);
\fill  (1.2,0) circle (1.2pt);
\fill  (0.8,0.63) circle (1.2pt);
\fill  (0.8,0.84) circle (1.2pt);
\fill  (0.8,1.12) circle (1.2pt);
\fill  (0.8,1.5) circle (1.2pt);
\fill  (0.8,-0.61) circle (1.2pt);
\fill  (0.8,-1.31) circle (1.2pt);

\fill  (-0.57,0) circle (1.2pt);
\fill  (-0.8,0) circle (1.2pt);
\fill (-1,0) circle (1.2pt);
\fill  (-1.2,0) circle (1.2pt);
\end{scriptsize}
\end{tikzpicture}
  \end{center}
  \caption{\label{fig:cad2}A cylindrical algebraic
      decomposition of $\set R^2$ induced by the polynomials in
      Example~\ref{ex:cad:2}.}
  \end{figure}
\end{example}

The full CAD algorithm works in three major steps. We start with a formula
$\phi$ of the form~\eqref{eq:form2} and collect the contained polynomials in a
set $P_n\subset\set R[x_1,\dots,x_n]$. The first step, the projection phase,
recursively adds new elements to $P_n$ such that its induced decomposition
becomes cylindrical. We denote this superset of $P_n$ by
$\operatorname{cadp}(P_n)$. If $n=1$, then $P_1$ is always cylindrical, so
$\operatorname{cadp}(P_1) := P_1$. For $n>1$, we compute a set $P_{n-1}$ which
contains all polynomials in $Q_n:=P_n\cap \set R[x_1,\dots,x_{n-1}]$ as well as
the image $P_n\setminus Q_n$ under a so called projection operator and return
$\operatorname{cadp}(P_n):=P_n\cup\operatorname{cadp}(P_{n-1})$. The projection
operator is a map such that $\operatorname{cadp}(P_n)$ is cylindrical if
$\operatorname{cadp}(P_{n-1})$ is. Intuitively it adds polynomials in
$\set R[x_1,\dots,x_{n-1}]$ to $P_{n-1}$ that correspond to asymptotes
orthogonal to the projection direction, intersections and self intersections of
the algebraic curves defined by the polynomials in $P_n\setminus Q_n$. In
Example~\ref{ex:cad:2}, $x\pm 1$ corresponds to the vertical asymptotes of the
algebraic curve given by $p_1$ and $x\pm c$ corresponds to the intersection of
the two curves given by $p_1$ and $p_2$.

In the second step, the extension phase,
sample points of the cells in the decomposition of $\set R$ induced by
$P_1$ are obtained by computing the roots of the polynomials in $P_1$
and points from the intervals between these roots. The cells of $\set
R$ are extended to cells of $\set R^2$ by keeping the $x_1$ values of
the sample points fixed and computing the roots of the polynomials in
$P_2$ regarded as univariate polynomials in $x_2$. This step is
iterated to obtain the cells in $\set R^3$, $\set R^4$ etc. In the last
step, the sample points of the cells in $\set R^n$ are plugged into
the the polynomials in $P$ and $\phi$ is evaluated.

It was shown by Brown and Davenport~\cite{brown} that the complexity of CAD is
double exponential in the number of variables. Many improvements of the base
algorithm like the ones found in~\cite{hong,marekbrown,chen}, however, allow
for solving moderately sized systems via CAD.

\subsection{Virtual Substitution}

The virtual substitution technique takes a more symbolic view on
the roots of a polynomial. It was introduced by Weispfenning in 1988,
see~\cite{weispfenning}, and several improvements and generalizations have been
developed since. It is not as prevalent as CAD due to its current degree
limitations in practice, but usually performs much better in terms of computing
time.

To get a good understanding of VS, consider first univariate polynomials and a
special form of the quantifier-free formula $\phi$ that contains no strict
inequalities but only Boolean combinations of expressions of the form
$p(x)\bowtie 0$ with $\bowtie\ \in\{\leq,=,\geq\}$. Similarly to CAD, VS
decomposes the space into connected cells.  However, while CAD does not really
exploit the literals but only the polynomials appearing in them, the cells in VS
are constructed such that the truth value of~$\phi$ (rather than the signs of
the images of the polynomials) remains invariant in each cell.

Let $p_1,p_2\in\set R[x]$ and $\phi= p_1\geq 0 \wedge p_2\geq 0$. The
real roots $r_1,\dots,r_k$ of $p_1$ given in ascending order decompose
$\set R$ into finitely many intervals
\[(-\infty,r_1],(r_1,r_2],\dots,(r_{k-1},r_k],(r_k,+\infty).\] The real roots of
$p_2$ then refine this decomposition such that in each interval, the truth
values of the inequalities and equations in $\phi$ do not change within an
interval. 

\begin{example}
\label{ex:vs:1}
Let $p_1= 10^{-1}(x + 5)(x + 2)(x - 6)$ and $p_2=x^2-9$ and $\Phi=\exists
x:p_1\geq 0\wedge p_2\leq 0$. Then the truth invariant decomposition induced by
the real roots of $p_1$ and $p_2$ consists of the intervals
\[(-\infty,-5],(-5,-3],(-3,-2],(-2,3],(3,6],(6,+\infty).\] By plugging in the
upper interval bounds (and evaluating the polynomials at~$+\infty$), we see that
$\phi\equiv \mathtt{true}$ via the test point $x=-3$.
\begin{figure}
  \begin{center}
    \begin{tikzpicture}[scale=0.4,line cap=round,line join=round,x=1.0cm,
      y=0.5244971462874863cm]
  \draw[->,color=black] (-7.39,0) -- (9.42,0); \foreach \x in {-6,-4,-2,2,4,6,8}
  \draw[shift={(\x,0)},color=black] (0pt,2pt) -- (0pt,-2pt);
  \draw[->,color=black] (0,-12.62) -- (0,6.45); \foreach \y in
  {-12,-10,-8,-6,-4,-2,2,4,6} \draw[shift={(0,\y)},color=black] (2pt,0pt) --
  (-2pt,0pt); \clip(-7.39,-12.62) rectangle (9.42,6.45); \draw[line width=1.1pt,
  smooth,samples=100,domain=-7.392955705201862:9.424128403621054]
  plot(\x,{1/10*((\x)+5)*((\x)+2)*((\x)-6)}); \draw[line width=1.1pt,
  smooth,samples=100,domain=-7.392955705201862:9.424128403621054]
  plot(\x,{(\x)^2-10});
\end{tikzpicture}

  \end{center}

  \caption{\label{fig:vs}Plot of the polynomials in Example~\ref{ex:vs:1}.}
\end{figure}
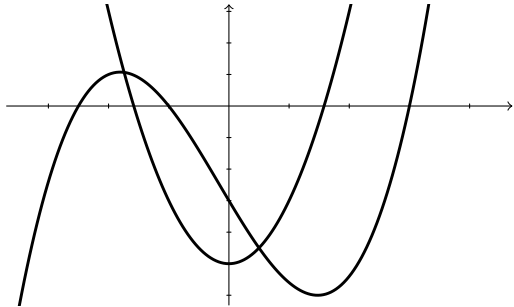
\end{example}
When dealing with multivariate polynomials in $\set R[x_1,\dots,x_n]$, the idea
is to choose one variable $x_i$ and view the polynomials as univariate in
$x_i$. Then we are in the univariate setting where we can (symbolically) compute
the interval decomposition. Here, the interval bounds are not real numbers but
expressions in the variables $x_1,\dots,x_{i-1},x_{i+1},\dots,x_n$.

\begin{example}
  Let $p_1=x_1x_2-1$ and $p_2=x_1-3$ and $\phi=\exists x_1\exists x_2: p_1\geq
  0\wedge p_2\leq 0$. As univariate polynomials in $\set R(x_1)[x_2]$, $p_2$
  either vanishes identically or has no roots. The polynomial $p_1$ has either
  no roots or a root at $x_1^{-1}$. We substitute this root expression for $x_2$
  and get
  \[p_1(x_1^{-1}/x_2) = x_1x_1^{-1}-1 = 0,\quad
  p_2(x_1^{-1}/x_2)=x_1-3.\] This substitution is only possible if we
  require that $x_1\neq 0$. Therefore, after the substitution, $\phi$
  becomes
  \[\exists x_1: 0\geq 0 \wedge x_1-3\leq 0 \wedge x_1\neq 0,\]
  and one quantifier has been removed. Continuing the process will give
  $\Phi\equiv\mathtt{true}$ via the test point $(3,\frac{1}{3})$.
\end{example}

In the example, the root expression has to be substituted into all polynomial
constraints, but it is also necessary to ensure that the substitution term is
valid. Here, this is achieved by adding a constraint to the formula to prevent
division by zero. Such additional constraints are called \textit{guards} of the
substitution term.  Also, substitution in the above example generates a
(quantified) Boolean combination of polynomial constraints; this is not always
the case.  Indeed substitution can lead for instance to rational functions.  In
virtual substitution, this problem is circumvented by a more sophisticated
substitution process.

Assume that after the substitution the resulting formula contains a relation of
the form $p/q \bowtie 0$ with $p$ and $q$ coprime polynomials in $\set
R[x_1,\dots,x_k]$. In order to remove the denominator, we can multiply the
relation by $q$. We do not know, however, if in the subsequent substitution
steps we derive values for $x_1,\dots,x_k$ such that $q$ would evaluate to a
strictly positive or negative number and thus whether the relation sign
$\bowtie$ changes or not.  Note that guards prevent $q$ to be zero.  A way out
is to multiply by $q^2$ (which is certainly positive) rather than $q$.

\begin{example}
In the formula 
  \[\exists x_1\exists x_2:  x_1x_2-1\geq 0\mathrel{\wedge}
  x_2 + x_1-3\leq 0.\] 
we substitute $x_2$ by $x_1^{-1}$ via virtual substitution and obtain the
equivalent formula
  \[\exists x_1: x_1 + x_1^2(x_1-3)\leq 0\mathrel{\wedge} x_1\neq 0.\]
\end{example}

In the full VS algorithm, several other substitution rules are necessary to
avoid non-polynomial expressions. These are detailed in~\cite{weispfenning} for
virtual substitution for polynomials of degree at most two.  Also included are
rules that allow strict inequalities by substitution of $\epsilon$-terms.  In
theory, the method can be extended to an arbitrary but fixed degree bound,
see~\cite{weispfenning3}, but there are still obstacles to overcome for higher
degree implementations.

Virtual substitution performs significantly better in theory and practice
compared to CAD. As shown in~\cite{weispfenning4}, VS is double exponential in
the number of quantifier alternations but only single exponential in the number
of quantified variables for a fixed quantifier type. Since the input in the SMT
setting does not contain quantifier alternations, virtual substitution is
significantly better compared to cylindrical algebraic decomposition for these
formulas in terms of theoretical complexity.

\section{Finding Conflict Sets}
\label{sec:ic}

In order to benefit from the interplay between SAT-solvers and special theory
solvers, it is required from the theory solver to provide small conflict sets.
The input to the theory solver is a conjunction of literals and if this
conjunction is not satisfiable, an answer in the form of a (hopefully small)
subset of the input literals that is unsatisfiable itself should be returned.
We call this answer a conflict set.  Such a conflict set should ideally be as
small as possible.  A minimum conflict set is a conflict set with minimum size,
whereas a minimal conflict set does not contain unnecessary literals, that is,
all its subsets are satisfiable.  A minimum conflict set is minimal, but a
minimal conflict set might not have the smallest size.  The procedure here is
not guaranteed to produce minimum or even minimal conflict sets, but we will
show in Section~\ref{sec:exp} that it is efficient at finding small conflict
sets.  We now describe how virtual substitution and cylindrical algebraic
decomposition can be adapted to provide such answers.

\subsection{Conflict Sets  and Linear Programming}
\label{sec:ic:lp}

The problem can be stated as follows: given an unsatisfiable quantified
formula~$\phi$ of the form
\begin{equation}
\label{eq1}
\phi=\exists x_1\dots\exists x_n:\bigwedge_{1\leq i\leq
  m}p_i\mathrel{\bowtie_i} 0,
\end{equation}
with $p_i\in\set R[x_1,\dots,x_n]$ and $\bowtie_i\ \in\{<, \leq, =, \neq, >,
\geq\}$, find a subset $I\subset\{1,\dots,m\}$ as small as possible such that
the formula
\[\phi'=\exists x_1\dots\exists x_n:\bigwedge_{i\in I}p_i\bowtie_i 0,\]
is unsatisfiable.\goodbreak

As was stated in the beginning of Section~\ref{sec:qe}, virtual substitution and
cylindrical algebraic decomposition share the same basic idea of finding a
finite set~$T$ of test points that suffice to determine the unsatisfiability of
$\phi$. The key to the problem of finding a conflict set is a reformulation of
the problem in terms of these test points.  For that, denote by $r_i$ the $i$th
polynomial constraint in $\phi$ for $i\in\{0,\dots,m\}$ and for each $i$ let
$e_i:T\rightarrow \{0,1\}$ be such that $e_i(a)=0$ if $r_i$ holds at $a$ and $1$
otherwise.  Applying CAD or VS to $\phi$ will result in $T=\{t_1,\dots,t_k\}$
such that for each $t\in T$ there exists an $i$ with $e_i(t)=1$.  Now let $v_i$
be the vector $(e_i(t_1),e_i(t_2),\dots,e_i(t_k))$.  Then the problem of finding
the smallest conflict set can be restated as a linear optimization
problem.\footnote{Alternatively, since $e_i(t_k)$ is either 0 or 1 for each $i$
  and $k$, the problem can be recast into propositional logic, and reduces then
  to finding the smallest implicant of a set of clauses, that is, the smallest
  set of literals implying all clauses.}  Considering a vector $w\in\{0,1\}^m$,
it is indeed equivalent to minimizing $w_1 + \dots + w_m$ under the linear
constraints \[Mw\geq \mathbf{1},\]
where $M$ is the matrix that contains the $v_i$ as columns and
$\mathbf{1}=(1,\dots,1)$. We will refer to matrices $M$ constructed in this way
as \textit{evaluation matrices}. If the vector~$w$ is as desired, then an entry
$1$ at the $i$th position means that $r_i$ is part of the computed conflict set.

Note that our reformulation yields a 0-1-linear integer programming
problem of the form
\begin{equation}
\label{eq2}
\min_{bw}\{w\in\set\{0,1\}^m\mid Mw\geq \mathbf{1}\},\text{ with
}b=\mathbf{1}=(1,\dots,1),M\in\{0,1\}^{k\times m},
\end{equation}
and we can use highly optimized linear programming techniques to find
an optimal or approximate solution.

This is only one of the benefits that the reformulation provides
us. Another one is that the information necessary to construct the
matrix $M$, i.e. the test points and images under the evaluation
functions $e_i$, is already computed during the quantifier
elimination. We will further investigate this fact in the next section. 

We can easily deduce that solving the linear optimization problem is not harder
than solving the original minimum conflict set problem:

\begin{theorem}
  Let $\mathcal{A}$ be an algorithm that solves the problem of finding a minimum
  conflict set. Then there exists a polynomial time algorithm $\mathcal{B}$ that
  transforms a matrix with entries in $\{0,1\}$ into a system of polynomials
  such that $\mathcal{A}\circ\mathcal{B}$ is an algorithm for solving linear
  optimization problems of the form~\eqref{eq2} 
\end{theorem}

\begin{proof}
  For a given matrix $M\in\{0,1\}^{k\times m}$, we show how to construct an
  equivalent conflict set problem in polynomial time, i.e. a formula
  $\phi$ whose minimum conflict set immediately yields a solution to the linear
  programming problem~\eqref{eq2}. Let $\phi$ be the quantified formula given by
  \[\phi=\exists x: \bigwedge_{i\in \{1,\dots,m\}}p_i=0,\]
  with
  \[p_i=\prod_{j=1}^{k}(x-j)^{1-M(j,i)}.\]
  One can easily check that the indices of the constraints in any minimum
  conflict set give rise to a solution of the linear programming
  problem. Multiplication of polynomials can be done in polynomial time, which
  proves the claim.\qed
\end{proof}

\subsection{Conflict Sets and Quantifier Elimination Optimization}

One of the main reasons why CAD and VS perform reasonably fast in practice is
that since their initial development, many improvements have been made to speed
up the computation.  For CAD, many of these improvements take the form of
specialized projection operators that reduce the number of cells that are
constructed in the projection phase for certain kinds of input. Another major
contribution was the development of \textit{partial cylindrical algebraic
  decomposition} by Collins and Hong in~\cite{hong}.  In the case of virtual
substitution, many improvements focus on the simplification of the quantifier
free formula after every substitution step. Most notably, this includes the work
of Sturm and Dolzmann in~\cite{dolzmann,sturm}.

While some of the improvements do not have an effect on the computation of
conflict sets as presented in Section~\ref{sec:ic:lp}, others will reduce the
amount of available information for the evaluation matrix. There are basically
two scenarios for information loss, which we describe with the help of two
showcase improvements for CAD and VS.

In the partial CAD method, the following rule is used to avoid unnecessary cell
construction. Note that we do not state it in full generality but adapt the rule
to our framework.

\begin{quote}
  Let $\phi$ be of the form~\eqref{eq1} with polynomials in $\set
  R[x_1,\dots,x_n]$. If $p\in\set R[x_1,\dots,x_k]$ appears in $\phi$
  with $k<n$ and there is a cell $C$ in the CAD of $\set R^{k}$
  induced by the polynomials in $\phi$ in which one of the constraints
  depending only on $p$ evaluates to $\mathtt{false}$, then the cells
  above $C$ do not have to be constructed.
\end{quote}

Assume $(a_1,\dots,a_k)\in\set R^k$ lies in such a cell with a
constraint containing~$p_i$ evaluating to $\mathtt{false}$ and further
assume we compute the CAD without the aforementioned rule. This means
that in the evaluation matrix we get $\ell$ rows corresponding to test
points $(a_1,\dots,a_k,*,\dots,*)$ with $\ell\geq 1$ and all entries
of the $i$th column are equal to 1 at the positions of these rows.
On the other hand, if we compute the partial CAD, these rows will be
missing in the evaluation matrix. However, we can add one row that
corresponds to the test point $(a_1,\dots.a_k)$ and we know that it
will contain at least one non-zero entry at position $i$. At
positions that correspond to polynomial constraints in more than the
first $k$ variables we insert the value $0$.
With this strategy, we can compensate for missing rows in the
evaluation matrix. It is important to note that in this setting, we do
not necessarily get a minimal conflict set even if we look for an
optimal solution in~\eqref{eq2}.

A second reason for missing information can be found in the
simplification strategies used in virtual substitution. If these
strategies can determine at some point in the computation that the
current quantifier-free formula (obtained for instance after some
substitution steps) is a tautology or a contradiction, the remaining
variables will not be substituted in the current substitution
branch. An example for such a situation is a formula of the form
\[x_k\geq0\wedge\dots\wedge x_k<0\wedge\dots\]
which is obviously a contradiction and instead of continuing the substitution
process, one can return $\mathtt{false}$ for this substitution branch.
This scenario is similar to the one before in that an unknown number of
rows in the evaluation matrix is missing. In contrast to the partial
CAD improvement however, the truth value of the substitution branch is
derived not from a single constraint but from a subset of the
constraints in the formula.

In order to preserve compatibility with the conflict set computation, we
therefore require that the simplification mechanism itself is able to determine
a \textit{local} conflict set, i.e.\ a conflict set of the quantifier-free
formula on which the simplification mechanism acts. We then can extend this to a
\textit{global} conflict set. The global conflict set should contain the union
of all the local conflict sets and the corresponding columns can be removed from
the evaluation matrix, together with all rows where these columns have non-zero
entries.

\section{Finding Conflict Sets via Redlog}
\label{sec:exp}

We implemented our method in the package Redlog, part of the open source
computer algebra system Reduce~\cite{reduce}.  We have adapted the available CAD
and VS implementations as well as parts of the simplification facilities for
quantifier-free formulas to explicitly provide the test point evaluations and
local conflict sets. Our method is such that it requires only little changes to
the highly optimized Redlog code. In other methods, see e.g.~\cite{smtrat}, the
implementations of CAD and VS are built from the ground up for use in SMT
solving.

To provide a reasonably large and meaningful test set, we used the
quantifier-free real arithmetic (QF\_NRA) benchmarks from the SMT-LIB library.
Our method expects a set of literals as input, so we use the veriT SMT-solver to
generate, for each SMT-LIB benchmark, one complete assignment of atoms in the
formula.  This assignment is satisfiable in the theory of real linear arithmetic
considering multiplication as an uninterpreted predicate.  This set is further
simplified using a preprocessor (which would eventually also have to be
considered in the conflict clause production).  This preprocessor only does
trivial rewriting.  Since Redlog is a generic tool and is not tuned for SMT-LIB
like formulas, it greatly benefits from this simple cleaning phase.  Finally,
among the obtained formulas, some are satisfiable, and are not considered here.
The test set thus obtained contains 6076 formulas that are proved unsatisfiable
by Redlog.  Figure~\ref{fig:exp1} provides an idea of the size of formulas: a
point $(x, y)$ on the curve means that there are $x$ formulas with a size
smaller than $y$.  The benchmarks as well as a distribution of Redlog featuring
conflict set computation can be obtained on
\url{http://www.loria.fr/~pdobal/}.\footnote{7947 formulas are provided,
  including the ones with a satisfiable or unknown status.}  All our experiments
use a 600 seconds timeout on a computer with an Intel i7-4600U CPU at
2.10GHz and 16 GB of RAM running Linux.

\begin{figure}
  \centering
    \includegraphics[scale=.7]{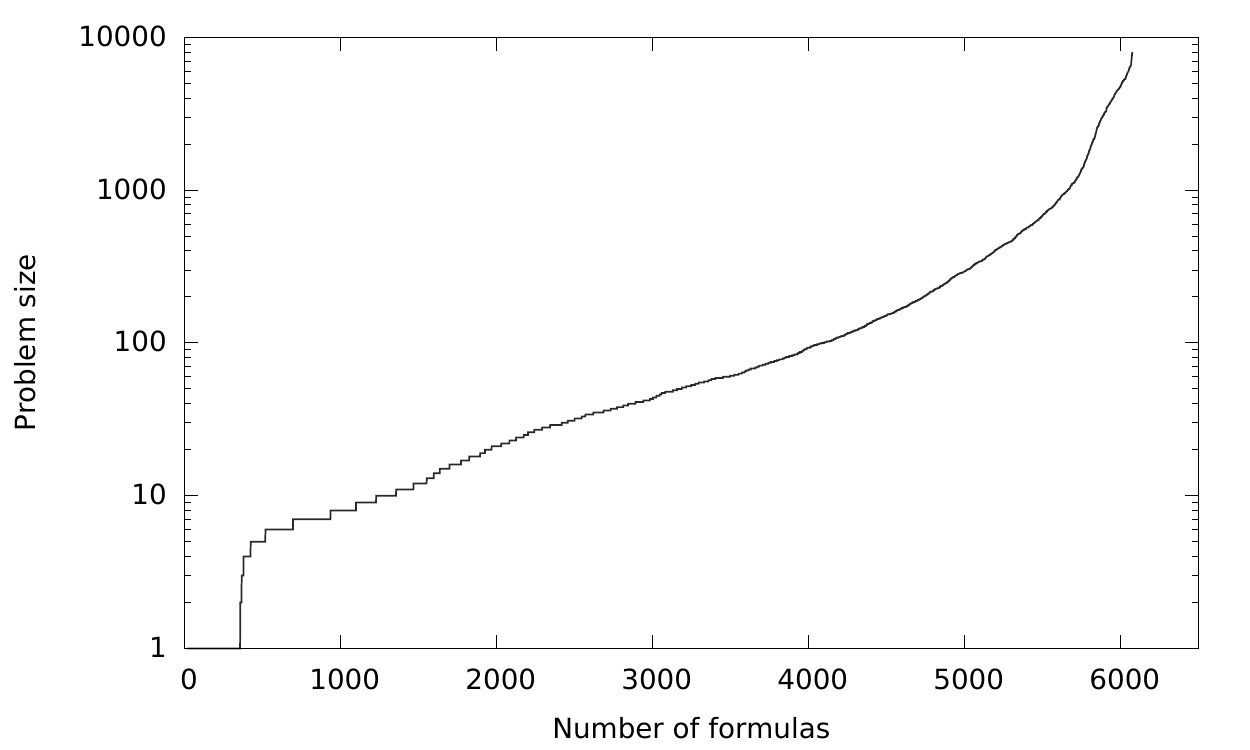}
  \caption{\label{fig:exp1}Problem size (in number of constraints) repartition.}
\end{figure}

\begin{figure}
  \centering
    \includegraphics[scale=.7]{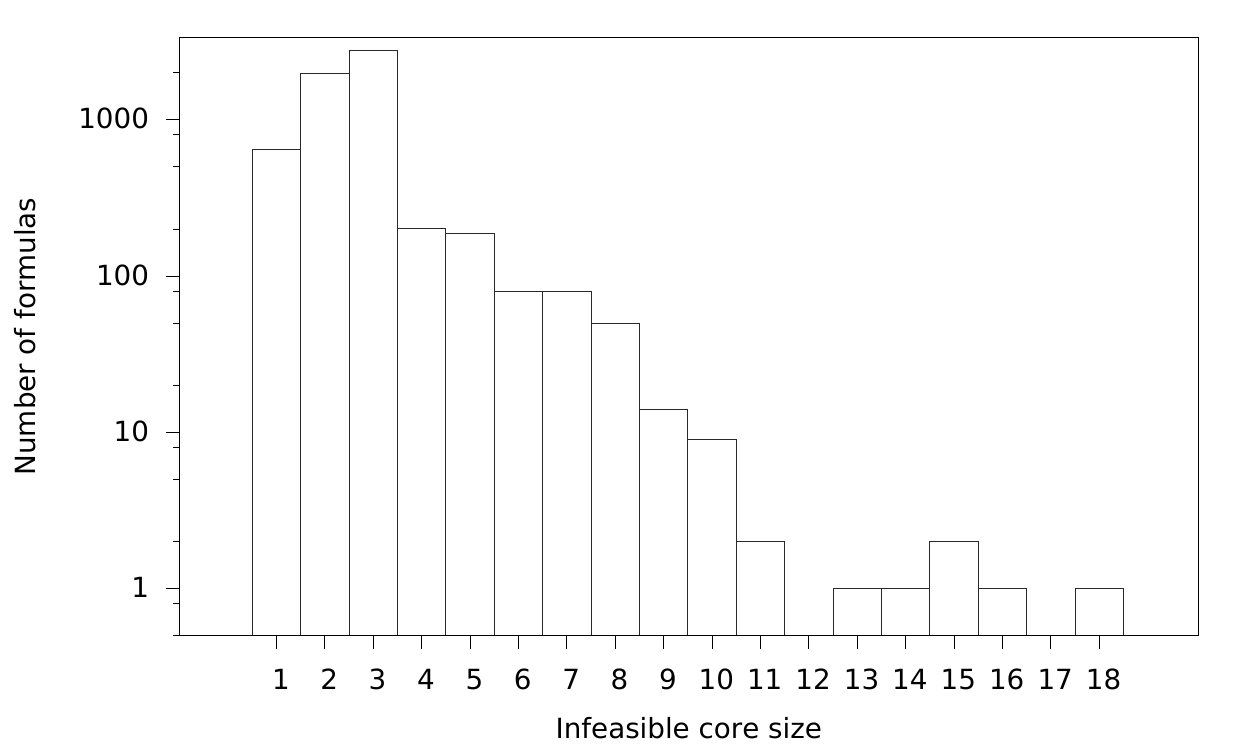}

    \caption{\label{fig:exp2}Number of formulas for a given conflict set size
      (in number of constraints).}
\end{figure}

\begin{figure}
  \centering
    \includegraphics[scale=.7]{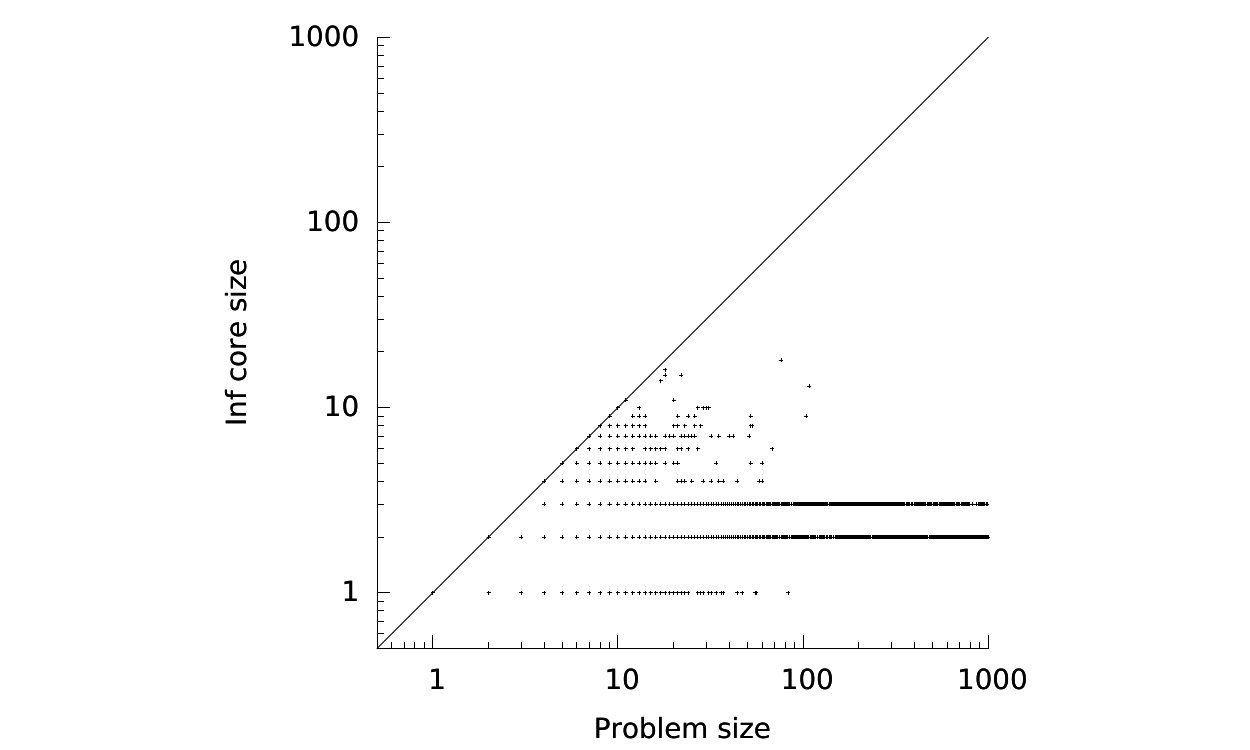}
    \caption{\label{fig:exp3}Size of formulas vs.\ size of conflict sets (in
      number of constraints).}
\end{figure}

The scatter plot on Figure~\ref{fig:exp3} gives a comparative view of the
problem and conflict set sizes, whereas Figure~\ref{fig:exp2} provides the
distribution of the conflict set sizes: the method is suitable to provide small
conflict sets.  Even if most inputs contain tens or hundreds of constraints,
just a few conflict sets have more than ten constraints.  Semiautomatic
inspection of the conflict sets exhibits that some of these are not minimal,
i.e.\ they contain literals that are not necessary for unsatisfiability.  For
integration within SMT, it will be necessary to evaluate whether it is more
efficient to reduce the conflict set size using other techniques or to keep
these perfectible conflict sets as they are.

\begin{figure}
  \centering
    \includegraphics[scale=.7]{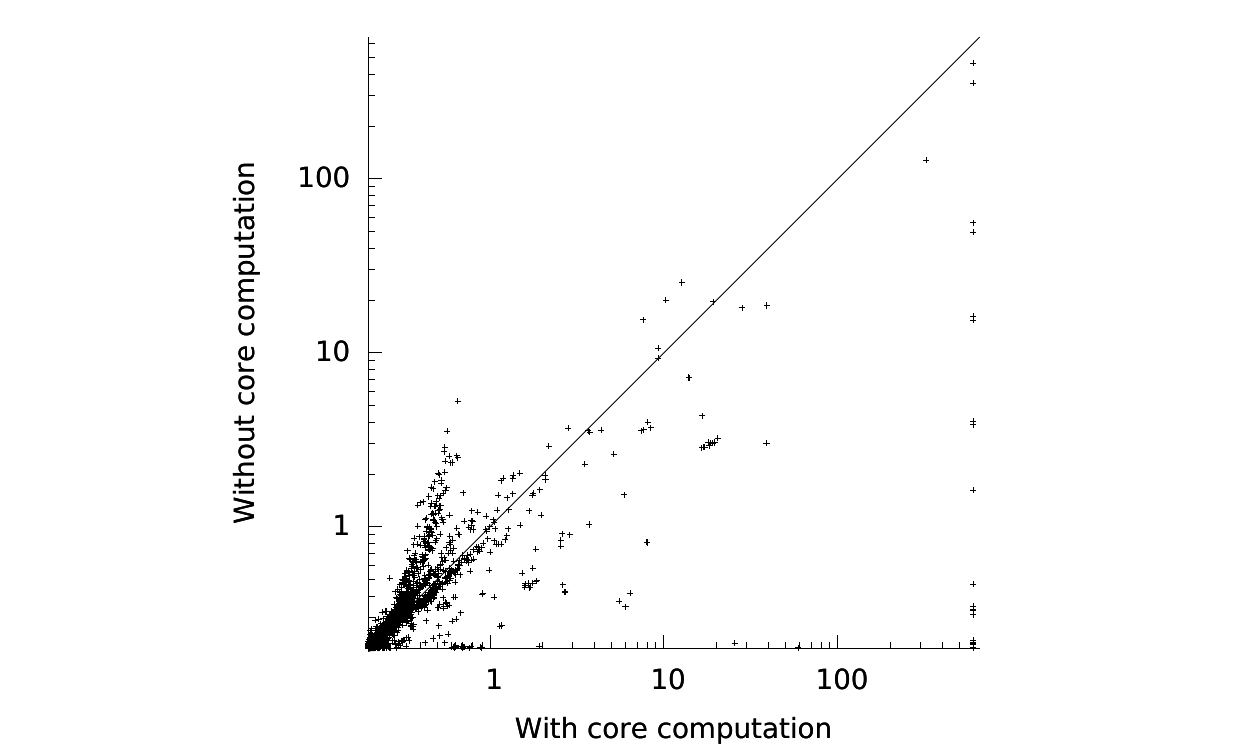}
  \caption{\label{fig:exp4}Computing time (in seconds) with and without conflict
    set generation.}
\end{figure}

Figure~\ref{fig:exp4} provides a comparative graph of the running times of
Redlog with and without conflict set generation.  Conflict set generation is not
exactly the non-conflict set producing algorithm with an additional phase: some
features of the original (non-conflict set producing) algorithm have to be
turned off.  This explains most of the cost, as well as the fact that sometimes
the conflict set  generating algorithm is faster (just because the search tree is
different). However the results clearly show that conflict set computation has an
acceptable cost; it fails only for 22 out of 6076 cases.

As a side note, Redlog was also evaluated against Z3 on all these benchmarks.
Redlog is definitely slower on most of them, also because there is a 0.2 seconds
cost for starting the whole Reduce infrastructure, whereas Z3 most of the
time answers in a few hundreds of a second.  It also appears that Z3 is
extremely effective for satisfiable files, being able to decide the
satisfiability of 24 files more than Redlog, whereas no file was stated
satisfiable by Redlog and not by Z3.  On the unsatisfiable problems, Redlog
succeeded on 2 among the 9 for which Z3 failed, whereas Redlog failed on 18
problems proved unsatisfiable by Z3.  This is an indication that further work to
present the SMT assignments to Redlog in a better way could lead to good results
when using Redlog as a back-end.

\section{Conclusion}

We introduced here a technique to adapt two commonly used real quantifier
elimination methods, that is, cylindrical algebraic decomposition and virtual
substitution, to also provide, besides the satisfiability status of a set of
polynomial constraints on the reals, a conflict set when the input set is
unsatisfiable.  This technique is based on the simple, yet effective,
observation that both methods amount to checking the values of the constraints
on a finite number of test points.  Collecting the test points and the values is
sufficient to compute the conflict sets in a post-processing phase, which is
basically a linear optimization problem, or the computation of a (prime)
implicant for a set of clauses.  Experimental results show that this technique
performs very well to produce small conflict sets.

Quantifier elimination methods also come with their lot of heuristics, and these
are not all seamlessly compatible with our technique.  Here, some of those
heuristics were turned off, and some were adapted to tag the constraints used by
the heuristics as mandatory for the conflict set.  This is responsible for
non-minimality of the produced conflict sets.  Although we can observe
experimentally that the produced conflict sets are small, it will certainly be
beneficial to better analyze the heuristics for finer conflict set production.

In their applications, SMT solvers are used to check large and mostly easy
computer generated formulas, whereas Redlog was mostly conceived for hard
problems of moderate size.  In order to succeed the integration of Redlog as a
complete back-end for non-linear constraints within SMT, it is necessary to
improve the heuristic simplification preprocessing phase, which is currently
extremely basic.  Another non-trivial issue is to take into account this
preprocessing phase for the conflict computation.\newline

\noindent\textbf{Acknowledgements.} We would like to thank the reviewers for
their valuable suggestions and comments on this paper. Furthermore, the
expertise of Thomas Sturm and Marek Ko\v{s}ta on Redlog was of much benefit to
the authors.

\bibliographystyle{abbrv}
\bibliography{references}
\end{document}